\documentclass[onefignum,onetabnum]{siamonline171218}

\usepackage{amsfonts}
\usepackage{amsmath}
\usepackage{array}
\usepackage{booktabs}
\usepackage{caption}
\usepackage{color}
\usepackage{graphicx}
\usepackage{microtype}
\usepackage{multirow}
\usepackage{nicefrac}
\usepackage{subfigure}
\usepackage{fancyvrb}
\usepackage{rotating}
\usepackage{wrapfig}

\usepackage{dgleich-math}
\usepackage{dgleich-tensor}

\newcommand{\hypergraph}{\mathcal{H}}

\newtheorem{example}[theorem]{Example}

\definecolor{mylinkcolor}{RGB}{0,0,140}
\hypersetup{colorlinks,allcolors=mylinkcolor,citecolor=mylinkcolor}
\usepackage[capitalize]{cleveref}
\crefname{section}{Section}{Sections}

\newcommand{\dataset}[1]{\textsc{#1}}
\newcommand{\ngrams}{\dataset{N-grams}}
\newcommand{\tags}{\dataset{tags-ask-ubuntu}}
\newcommand{\dawn}{\dataset{DAWN}}

\usepackage{paralist}

\renewenvironment{enumerate}[1]{\begin{compactenum}#1}{\end{compactenum}}

\usepackage[square,semicolon]{rnatbib}
\renewcommand{\cite}{\citep}

\usepackage{tikz,nicefrac,amsmath,pifont,tkz-graph,makecell,tkz-graph}
\usetikzlibrary{arrows,backgrounds,patterns,matrix,shapes,fit,calc,shadows,plotmarks}
\usetikzlibrary{arrows.meta}

\begin{document}

\VerbatimFootnotes

\title{Three hypergraph eigenvector centralities}
\author{
Austin R.~Benson \\
Cornell University \\
arb@cs.cornell.edu
}

\maketitle

\begin{abstract}
Eigenvector centrality is a standard network analysis tool for determining the importance of 
(or ranking of) entities in a connected system that is represented by a graph.
However, many complex systems and datasets have natural multi-way interactions that are more faithfully
modeled by a hypergraph.
Here we extend the notion of graph eigenvector centrality to uniform hypergraphs.
Traditional graph eigenvector centralities are given by a positive eigenvector of the adjacency matrix,
which is guaranteed to exist by the Perron-Frobenius theorem under some mild conditions.
The natural representation of a hypergraph is a hypermatrix (colloquially, a tensor).
Using recently established Perron-Frobenius theory for tensors, 
we develop three tensor eigenvectors centralities for hypergraphs, each with different interpretations.
We show that these centralities can reveal different information on real-world data by analyzing hypergraphs constructed from n-gram frequencies, co-tagging on stack exchange, and drug combinations observed in patient emergency room visits.
\end{abstract}

\begin{keywords}
  hypergraph, tensor, eigenvector, centrality, network analysis
\end{keywords}


\section{Finding important entities from relations}
The central question of centrality and ranking in network analysis is:\
how do we find the important entities given a set of relationships between them?
Make no mistake---when the relationships are pairwise and the system is modeled by a graph, 
there is a plethora of definitions and methods for centrality~\cite{Borgatti-2006-centrality,Estrada-2010-properties,Boldi-2014-axioms},
and the study of centrality in social network analysis alone has a long-standing history~\cite{Bavelas-1950-communication,Katz-1953-status,Sabidussi-1966-centrality,Bonacich-1972-eigenvector,deSolaPool-1978-contacts,Freeman-1977-betweenness}.
Somewhat more modern developments come from Web applications, such as PageRank, which was
used in the early development of Google search results~\cite{Brin-1998-anatomy,Page-1999-PageRank}, 
and hub and authority scores, which were used to find authoritative Web sources~\cite{Kleinberg-1999-HA}.
Centrality measures are a pivotal part of the network analysis toolbox and thus get used in a variety of 
applications~\cite{Gleich-2015-PageRank,Jeong-2001-lethality,Bullmore-2009-brains}.
And in addition to the problem of identifying important nodes, centrality is also used as a
feature in network analysis machine learning tasks such as
role discovery~\cite{Henderson-2012-RolX},
computing graph similarity~\cite{Koutra-2013-DeltaCon},
and spam detection~\cite{Ye-2015-spam}.

A major shortcoming of network centrality stems from the long-running assumption throughout network science
that relationships are pairwise measurements and hence a graph is
the appropriate mathematical model~\cite{Strogatz-2001-exploring,Newman-2003-survey}.
Thus, nearly all centrality measures are designed within this dyadic paradigm.
However, many systems contain simultaneous interactions between several entities.
For example, people communicate and collaborate in groups,
chemical reactions involve several reagents, and
more than two labels may be used to classify a product.
In these cases, a hypergraph is a more faithful model,
but we lack foundational mathematical analogs of centrality for this model.

This paper focuses on developing analogs of graph \emph{eigenvector centrality} for hypergraphs.
The term ``eigenvector centrality'' has two meanings in network science. 
Sometimes, eigenvector centrality means any set of centrality scores on nodes
that is an eigenvector of some natural matrix associated with the network at hand.
These include, for example, the aforementioned 
PageRank\footnote{PageRank has been called the ``\$25,000,000,000 Eigenvector''~\cite{Bryan-2006-google}.}
and hub and authority scores.
Other times, eigenvector centrality refers specifically to the principal eigenvector of the adjacency matrix of a graph~\cite{Newman-2008-math}
(which makes the vernacular confusing).
This eigenvector centrality was originally proposed by \citet{Bonacich-1972-eigenvector}, was later
used to study social networks~\cite{Bonacich-1987-power}, and will be the notion of eigenvector centrality used in this paper.

\paragraph*{Background on graph eigenvector centrality}
Here we provide the requisite background on eigenvector centrality for graphs.
We generalize the formulation to hypergraphs in the next section.
Assume that we have a strongly connected (possibly directed)
graph $G = (V, E)$ with adjacency matrix $\mA$.
The eigenvector centrality $\vc$ may be derived via the following two desiderata~\cite{Bonacich-1972-eigenvector,Newman-2008-math}:
\begin{enumerate}
\item The centrality score of each node $u$, $c_u$, is proportional 
to the sum of the centrality scores of the neighbors of $u$, i.e., $c_u \propto \sum_{(u, v) \in E} c_v$.
\item The centrality scores should be positive, i.e., $\vc > 0$.
\end{enumerate}
Assuming the same proportionality constant, we may auspiciously
write the first condition as
\begin{equation}
c_u = \frac{1}{\lambda}\sum_{(u, v) \in E} c_j,\;\text{for all } u \in V,
\end{equation}
where $\lambda$ is a constant.
The matrix enthusiast quickly recognizes that $\vc$ is an eigenvector of $\mA$:
\begin{equation}\label{eq:evec_centrality}
\mA\vc = \lambda \vc.
\end{equation}
\Cref{eq:evec_centrality} holds for any eigenpair of $\mA$.
The second of our desiderata, along with the Perron-Frobenius theorem, tells us which one to use.

\begin{theorem}[Perron-Frobenius Theorem for matrices as in Theorem 1.4 of \cite{Berman-Plemmons-1994-book}]\label{thm:pf}
Let $\mA$ be an irreducible matrix.
Then there exists an eigenvector $\vc > 0$ such that $\mA\vc = \lambda_1 \vc$,
$\lambda_1 > 0$ is an eigenvalue of largest magnitude of $\mA$, the eigenspace associated with
$\lambda_1$ is one-dimensional, and $\vc$ is the only nonnegative eigenvector of $\mA$ up to scaling.
\end{theorem}
If $\mA$ is the adjacency matrix of a strongly connected graph, then $\mA$ is
irreducible and we can apply the theorem. The vector $\vc$ gives the centrality scores,
which are unique up to scaling.
Eigenvector centrality of this form has appeared in a range of
applications, including the analysis of infectious disease spreading in
primates~\cite{Balasubramaniam-2016-social}, patterns in fMRI data of human
brains~\cite{Lohmann-2010-eigenvector}, and career trajectories of Hollywood
actors~\cite{Taylor-2017-eig-temporal}.

\section{Hypergraph eigenvector centralities}\label{sec:defs}

Instead of a graph, we now assume that our dataset is an $m$-uniform 
hypergraph $\hypergraph = (V, E)$, which means that each hyperedge $e \in E$ 
is a size-$m$ subset of $V$.
If $m = 3$, a natural representation of 
$\hypergraph$ is an $n \times n \times n$ symmetric
 ``hypergraph adjacency tensor":%
\footnote{Technically, this object is a hypermatrix. However, ``tensor" is synonymous with
multi-dimensional array the data mining community~\cite{Kolda-2009-tensor-decompositions}, 
so we use it here. See \citet{Lim-2014-tensors} for precise distinctions.}
\begin{equation}\label{eq:adj_tensor}
\cT_{u, v, w} =  
\begin{cases}
1 & \text{if}\ (u, v, w) \in E \\
0 & \text{otherwise}.
\end{cases}
\end{equation}
When deriving graph eigenvector centrality above, we used an irreducible adjacency matrix from a strongly connected graph.
We need analogous notions for tensors and hypergraphs.

\begin{definition}[Irreducible tensor~\cite{Lim-2005-variational}]\label{def:irr}
An order-$m$, dimension-$n$ tensor $\cmT$ is \emph{reducible} if there exists a non-empty 
proper subset $S \subset \{1, \ldots, n\}$ such that  for any $i \in S$ and $j_2, \ldots, j_m \notin S$, $\cmT_{i, j_2, \ldots, j_m} = 0$.
If $\cmT$ is not reducible, then it is \emph{irreducible}.
\end{definition}

We introduce connected hypergraphs here using the language of tensors.
The definition is the same as classical notions of connectivity in hypergraphs~\cite{Berge-1984-hypergraphs} when the tensor is symmetric,
which is the case in \cref{eq:adj_tensor}.

\begin{definition}[Strongly connected hypergraph]\label{def:conn}
An $m$-uniform, $n$-node hypergraph with adjacency tensor $\cmT$ is
\emph{strongly connected} if the graph induced by the $n \times n$ matrix
$\mM$ obtained by summing the modes of $\cmT$, $M_{ij} = \sum_{j_3, \ldots, j_m} \cT_{i,j,j_3, \ldots, j_m}$, is strongly connected.
\end{definition}

The matrix $\mM$ defined above is called the \emph{representative matrix} of $\cmT$,
and, importantly, \emph{a strongly connected hypergraph has an irreducible adjacency tensor}~\cite{Qi-2017-tensor-book}.
Here we assumed an ``undirected" set-based definition of hypergraphs, so $\cmT$ is symmetric following \cref{eq:adj_tensor}.
This means that the graph induced by $\mM$ is undirected and ``strongly connected" really just means ``connected.''
Furthermore, the graph induced by $\mM$ has the same connectivity as the clique expansion graph of a hypergraph~\cite{Agarwal-2006-higher},
where each hyperedge induces a clique on the nodes in the graph.
There are natural notions of directed hypergraphs with non-symmetric adjacency tensors~\cite{Gallo-1993-directed-hgraph},
and some of the theorems we use later still apply in these cases.
Therefore, we use the term ``strongly connected" throughout.

In the rest of this section, we develop three eigenvector centralities for strongly connected hypergraphs.
To do so, we generalize the desiderata for the eigenvector centrality scores $\vc$:
\begin{enumerate}
\item Some function $f$ of the centrality of node $u$, $f(c_u)$, should be proportional 
to the sum of some function $g$ of the centrality score of its neighbors. In a 3-uniform hypergraph,
this means that for some positive constant $\lambda$,
\begin{equation}\label{eq:des1}
f(c_u) = \frac{1}{\lambda}\sum_{(u, v, w) \in E} g(c_v, c_w)
\end{equation}
\item The centrality scores should be positive, i.e., $\vc > 0$.
\end{enumerate}
Different choices of $f$ and $g$ give new notions of centrality.
Careful choices of $f$ and $g$ relate to matrix and tensor eigenvectors.
To keep notation simpler, we use $3$-uniform hypergraphs when
introducing new concepts (as in \cref{eq:des1})
and then generalize ideas to $k$-uniform hypergraphs.

\subsection{Clique motif Eigenvector Centrality (CEC)}

Perhaps the most innocuous choice of $f$ and $g$
in \cref{eq:des1} are $f(c_u) = c_u$ and 
$g(c_v, c_w) = c_v + c_w$. 
In this case, there is a simple matrix formulation of the eigenvector formulation. 
This is unsurprising since $f$ and $g$ are linear.

\begin{proposition}\label{prop:cec}
Let $\hypergraph$ be a strongly connected $3$-uniform hypergraph.
When $f(c_u) = c_u$ and $g(c_{v}, c_{w}) = c_v + c_w$ in \cref{eq:des1},
the centrality scores are given by the eigenvector of the largest real eigenvalue of the
matrix $\mW$, where $W_{ux}$ is the number of hyperedges containing $u$ and $x$.
\end{proposition}
\begin{proof}
\begin{align*}
\lambda f(c_u) = \lambda c_u = \sum_{(u, v, w) \in E} g(c_v, c_w)
= \sum_{(u, v, w) \in E} c_{v} + c_w
&= \sum_{e \in E \;:\; \{u, x\} \subset e} c_{x} = \sum_{x} W_{ux}c_x.
\end{align*}
Thus, $\lambda \vc = \mW \vc$, and we assumed above that $\lambda > 0$ and $\vc > 0$.
If $\hypergraph$ is strongly connected, then the undirected graph induced by $\mW$
is connected and $\mW$ is irreducible. Applying \cref{thm:pf} says that $\vc$ must
be the eigenvector corresponding to the largest real eigenvalue.
\end{proof}
The matrix $\mW$ was called the ``motif adjacency matrix" by the author in
previous work~\cite{Benson-2016-higher,Yin-2017-local}.
Specifically, it would be the triangle motif adjacency matrix, 
if you  interpret $3$-uniform hyperedges as triangles in some graph.
We give a formal definition for the general case.

\begin{definition}[Clique motif Eigenvector Centrality (CEC)]
Let $\hypergraph$ be a strongly connected $m$-uniform hypergraph.
Then the clique motif eigenvector centrality scores $\vc$ are given by the
eigenvector $\mW\vc = \lambda_1\vc$, where $\|\vc\|_1 = 1$, $W_{uv}$ is the number of hyperedges
containing nodes $u$ and $v$, and $\lambda_1$ is the largest real eigenvalue of $\mW$.
\end{definition}

One interpretation of CEC (and eigenvector centrality for undirected graphs in general) 
is via a best low-rank decomposition. 
Assuming that the graph induced by $\mW$ is non-bipartite (which it will be if $m \ge 3$,
since hyperedges induce cliques in $\mW$), then
$\lambda_1 > 0$ is the unique largest magnitude eigenvalue of the symmetric matrix $\mW$~\cite{Lovasz-2007-eigenvalues}, 
and $\vc$ is also the principal left and right singular vector of $\mW$. Thus, by the Eckart--Young--Mirsky theorem~\cite[Theorem 2.4.8]{Golub-2012-matrix},
$\vc \propto \argmin_{\|\vx\| \in \mathbb{R}^n} \| \mW - \vx\vx^T \|_F$.
We can also interpret CEC with averaged path counts.
First, observe that
\begin{align*}
\text{\#(length-$\ell$ paths to $u$)} 
= 
\sum_{(u, v, w) \in E} 
\text{\#(length-($\ell$-1) paths to $v$)} +
\text{\#(length-($\ell$-1) paths to $w$)}.
\end{align*}
Let $\vp^{(\ell)}$ be a vector of that counts the number of length-$\ell$ paths ending at each node
from any starting node. Then
\begin{align}
\textstyle p_u^{(1)}    &= \textstyle [\mW\ve]_u = \sum_{(u, v, w) \in E} e_v + e_w = \#((u,v,w) \in E)\; \text{and } \\
\textstyle p_u^{(\ell)} &= \textstyle [\mW\vp^{(\ell-1)}]_u = \sum_{(u, v, w) \in E} p_v^{(\ell-1)} + p_w^{(\ell-1)},
\end{align}
where $\ve$ is the vector of all ones.
If we think of the CEC vector $\vc$ as the limit of the power method algorithm, then $\vc$ can
be interpreted as the steady state of a weighted average of infinite paths through the hypergraph
(see \citet{Benzi-2015-limiting} for a more formal argument).

Computing the CEC vector is often straightforward.
If $\hypergraph$ is strongly connected,
then the undirected graph induced by $\mW$ is connected.
If this graph is also non-bipartite (which, again, must be the case for $m$-uniform hypergraphs when $m \ge 3$),
then the eigenvalue in \cref{thm:pf} is the unique eigenvalue of largest magnitude.
In this case, we can we can use the power method to reliably compute $\vc$.

One subtlety is that the eigenvector is only defined up to its magnitude.
Usually, this issue is ignored, under the argument that only relative order
matters for ranking problems. However, we should be conscientious when
centrality scores are used as features in machine learning.
For example, the scale of a centrality vector as a node feature would affect common tasks such
as principal component analysis, where scale changes variance.
(These issues can also be alleviated by pre-processing techniques, such
as normalizing features to have zero mean and unit variance, although such pre-processing
is not always employed.)
 In this paper, to make the presentation simple,
we assume that centrality vectors are scaled to have unit $1$-norm.

\subsection{$Z$-eigenvector centrality (ZEC)}

To actually incorporate non-linearity, we can keep the innocuous choice $f(c_u) = c_u$ 
but change $g$ to the product form: $g(c_v, c_w) = c_vc_w$ in \cref{eq:des1}.
Now, the contribution of the centralities of two nodes in a 3-node hyperedge is multiplicative for the third.
This leads to the following system of nonlinear equations for a 3-uniform hypergraph:
\begin{equation}\label{eq:z_evec_centr}
c_u = \frac{1}{\lambda} \sum_{(u, v, w) \in E} c_vc_w,\; u \in V \iff \cmT\vc^2 = 2\lambda \vc.
\end{equation}
Here, $\cmT\vc^2$ is short-hand for a vector with $[\cmT\vc^2]_i \equiv \sum_{j,k}\cT_{i,j,k}c_jc_k$
(similarly, for an order-$m$ tensor, $[\cmT \vc^{m-1}]_i \equiv \sum_{j_2, \ldots, j_m} \cT_{i,j_2,\ldots,j_m}c_{j_2} \cdots c_{j_m}$).
The extra factor of 2 comes from the symmetry in the adjacency tensor (for an order-$m$ tensor, this extra factor is $(m -1)!$).

A real-valued solution $(\vc, \lambda)$ with $\vc \neq 0$ to \cref{eq:z_evec_centr} 
is called a tensor $Z$-eigenpair~\cite{Qi-2005-Z-evec} or a tensor $l^2$ eigenpair~\cite{Lim-2005-variational} (we will use the ``$Z$'' terminology).
At first glance, it is unclear if such an eigenpair even exists, let alone a positive one.
Assuming the hypergraph is strongly connected, \citet{Chang-2008-Tensor-PF} 
proved a Perron-Frobenius-like theorem
that gives us the existence of a positive solution $\vc$.

\begin{theorem}[Perron-Frobenius for $Z$-eigenvectors---Corollary 5.10 of \cite{Chang-2008-Tensor-PF}%
\footnote{An erratum was published for this result, but the error does not affect our statement or analysis. 
See \citet[Theorem 2.6]{Chang-2013-Z} from the same authors for the corrected result.}%
]\label{thm:PFZ}
  Let $\cmT$ be an order-$m$ irreducible nonnegative tensor.
  Then there exists a $Z$-eigenpair
  $(\vx, \lambda_1)$ satisfying $\cmT\vx^{m-1}=\lambda_1\vx$ such that $\lambda_1>0$ and $\vx > 0$.
\end{theorem}
Unlike the case with matrices, there can be multiple positive $Z$-eigenvectors, even for the same eigenvalue~\cite[Example 2.7]{Chang-2013-Z}.
With this tensor Perron-Frobenius theorem in hand, we can define $Z$-eigenvector centrality for hypergraphs.
To manage the uniqueness issue, we consider any positive solution to be a centrality vector.

\begin{definition}[$Z$-eigenvector centrality (ZEC)]\label{def:zec_vec}
Let $\hypergraph$ be a strongly connected $m$-uniform hypergraph with adjacency tensor $\cmT$.
Then a $Z$-eigenvector centrality vector for $\hypergraph$ is \emph{any} positive vector $\vc$ 
satisfying $\cmT\vc^{m-1} = \lambda \vc$ and $\|\vc\|_1 = 1$ for some eigenvalue $\lambda > 0$.
\end{definition}

Analogous to the CEC (or standard graph) case, there is a ZEC vector that is a best low-rank approximation.
To prove this, we first need the following lemma.
\begin{lemma}\label{lem:positivity}
Let $\cmT$ be an irreducible symmetric nonnegative tensor and suppose that $\vx$ is a nonnegative $Z$-eigenvector
of $\cmT$ with positive eigenvalue $\lambda > 0$. Then $\vx$ is positive.
\end{lemma}
\begin{proof}
The proof technique follows \citet[Lemma 21]{Qi-2016-nonnegative}.
Since $\lambda > 0$ and $\vx \ge 0$, there must be some coordinate $i$ such that $x_i > 0$.
By \cref{eq:z_evec_centr} and nonnegativity of $\cmT$,
\begin{equation}
0 < \lambda x_i = \sum_{j_2, \ldots, j_m \colon \cT_{i, j_2, \ldots, j_m} > 0 } \cT_{i, j_2, \ldots, j_m} x_{j_2} \cdots x_{j_m} \implies x_{j_2}, \ldots, x_{j_m} > 0
\end{equation}
Therefore, $x_r > 0$ for any index $r \in \{(i, j_2, \ldots, j_m) \;\vert\; \cT_{i, j_2, \ldots, j_m} > 0\}$.
Iterating this argument shows that $x_s > 0$ for any index $s$ reachable from $i$ in the graph induced by 
the representation matrix
$M_{ij} = \sum_{j_2, \ldots, j_n} \cT_{i,j,j_2, \ldots, j_n}$. Since $\cmT$ is irreducible, this is all indices, so $\vx > 0$.
\end{proof}

The following theorem says that there the ZEC vector is proportional to a best rank-1 approximation vector of
the hypergraph adjacency tensor. However, neither ZEC vectors nor best rank-1 approximations 
need be unique~\cite{Friedland-2014-number}.
\begin{theorem}\label{thm:main}
Let $\hypergraph$ be an $m$-uniform strongly connected hypergraph with symmetric adjacency tensor $\cmT$.
Then there is a ZEC vector $\vc \propto \vv$, where 
$\vv \in \argmin_{\vx \in \mathbb{R}^n} \| \cmT - \otimes^m \vx \|_F$
and $\otimes^m \vx$ is the order-$m$ symmetric tensor $\cmS$ defined by 
$\cS_{i_1, \ldots, i_m} = x_{i_1} \cdots x_{i_m}$.
\end{theorem}
\begin{proof}
The proof combines several prior results on tensors with \cref{lem:positivity}.
First, any best \emph{symmetric} rank-1 approximation to a symmetric
tensor is a tensor $Z$-eigenvector with largest magnitude eigenvalue~\cite[Theorem 3]{Kofidis-2002-best}.
Second, the best rank-1 approximation to a symmetric tensor can
be chosen symmetric~\cite[Theorem 4.1]{Chen-2012-Opt}.
Thus, a best rank-1 approximation can be chosen to be the $Z$-eigenvector with largest magnitude eigenvalue.
Third, the coordinate values of any best rank-1 approximation to a nonnegative tensor can be chosen
so that its entries are nonnegative~\cite[Theorem 16]{Qi-2016-nonnegative}, so there is a nonnegative
eigenvector with largest magnitude eigenvalue.
Fourth, the largest $Z$-eigenvalue in magnitude is positive~\cite[Theorem 3.11 and Corollary 3.12]{Chang-2013-Z}.
Finally, \cref{lem:positivity} says that the corresponding eigenvector must be positive.
\end{proof}

Computing tensor $Z$-eigenvectors is much more challenging than computing matrix eigenvectors;
computing a best symmetric rank-1 approximation to a tensor is NP-hard~\cite[Theorem 10.2]{Hillar-2013-nphard}.
Adjacency tensors of hypergraphs are symmetric tensors, so we might first try the symmetric higher-order power method,
an analog of the power method for matrices; however, such methods are not guaranteed to converge~\cite{DeLathauwer-2000-best,Regalia-2000-hopm,Kofidis-2002-best}.
A shifted symmetric higher-order power method with an appropriate shift guarantees convergence to 
\emph{some} $Z$-eigenpair~\cite{Kolda-2011-sshopm,Kolda-2014-GEAP}
but can only converge to a class of so-called ``stable'' eigenpairs.\footnote{Let
$(\lambda, \vx)$ be an eigenpair of an order-$m$ symmetric tensor $\cmT$ with $\| \vx \|_2 = 1$ and $\mU$
be an orthonormal basis of the subspace orthogonal to $\vx$. Then  the eigenpair is \emph{unstable} if
$\mU^T((m-1)\cmT[\vx] - \lambda \vx)\mU$ is indefinite, where
 $\cmT[\vx]$ is the matrix with $(i, j)$ entry $\sum_{j_3, \ldots, j_n} T_{i,j,j_3,\ldots,j_n} x_{j_3} \cdots x_{j_n}$. 
 An eigenpair is stable if it is not unstable.}
It turns out that ZEC vectors can be unstable, which hinders our reliance on these algorithms.

\begin{example}\label{ex:sshopm_counter}
The following strongly connected 3-uniform hypergraph has a ZEC vector $\vc$ that is
an unstable eigenvector in the sense of \citet{Kolda-2011-sshopm}:
\begin{center}\includegraphics[height=2cm]{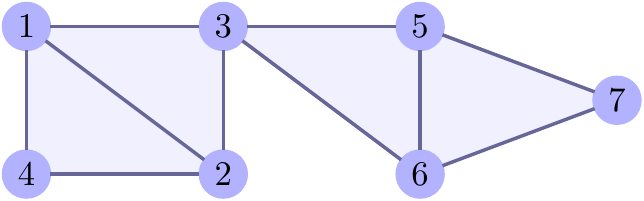}\end{center}
Indeed, one can verify that $(\vx, \sqrt{2})$ is a $Z$-eigenpair, where $\| \vx \|_2 = 1$ with
$x_1 = x_2 = x_5 = x_6 = \sqrt{6} / 6$, $x_4 = x_7 = \sqrt{2} / 6$, and $x_3 = \sqrt{2} / 3$.
Some simple calculations following \citet[Definition 3.4]{Kolda-2011-sshopm} show that $\vx$ is
an unstable $Z$-eigenvector.
\end{example}

There are algorithms based on semi-definite programming hierarchies that are guaranteed to
compute the eigenvectors~\cite{Cui-2014-real,Nie-2014-sdp,Nie-2017-eigenvalues}, but these
methods do not scale to the data problems we explore in \cref{sec:data}.
Recent work by the author develops a method to compute $Z$-eigenpairs using dynamical systems,
which can scale to large tensors and also compute unstable eigenvectors~\cite{Benson-2018-Z}, 
albeit without theoretical guarantees on convergence. 
In fact, we used this method to discover the example in \cref{ex:sshopm_counter}.
We use this algorithm for our computational experiments.

\subsection{$H$-eigenvector centrality (HEC)}

A reasonable qualm with ZEC is that the dimensional analysis is nonsensical---if centrality 
is measured in some ``unit,'' then \cref{eq:z_evec_centr} says that a unit of centrality is equal 
to the sum of the product of that same unit.
With this in mind, we might choose $f(c_u) = c_u^2$ and $g(c_v, c_w) = c_vc_w$ in 
\cref{eq:des1} to satisfy dimensional analysis:
\begin{equation}\label{eq:h_evec_centr}
c_u^2 = \frac{1}{\lambda} \sum_{(u, v, w) \in E} c_vc_w,\; u \in V \iff \cmT\vc^2 = 2 \lambda \vc^{[2]}.
\end{equation}
Here, $\vc^{[k]}$ is short-hand notation for the entry-wise $k$th power of a vector.\footnote{Written as \verb+c .^ k+ in Julia or MATLAB.}
Again, the extra factor of 2 comes from the symmetry in the adjacency tensor.

A real-valued solution $(\vc, \lambda)$ to \cref{eq:h_evec_centr} with $\vc \neq 0$
is called a tensor $H$-eigenpair~\cite{Qi-2005-Z-evec} or a tensor $l^k$-eigenpair~\cite{Lim-2005-variational} (we will use the ``$H$'' terminology).
Again, we can employ tensor Perron-Frobenius theory for the existence of a positive solution with positive eigenvalue.

\begin{theorem}[Perron-Frobenius for $H$-eigenvectors---Theorem 1.4 of \cite{Chang-2008-Tensor-PF}]\label{thm:PFH}
  Let $\cmT$ be an order-$m$ irreducible tensor. Then there exists an $H$-eigenpair
  $(\vx, \lambda_1)$ with $\vx > 0$ and $\lambda_1 > 0$.
  Moreover, any nonnegative $H$-eigenvector also has eigenvalue $\lambda_1$,
  such vectors are unique up to scaling, and $\lambda_1$ is the largest eigenvalue in magnitude.
\end{theorem}

The result is stronger than for $Z$-eigenvectors (\cref{thm:PFZ})---the positive $H$-eigenvector is unique up to scaling.
With this result, we define our third hypergraph eigenvector centrality.

\begin{definition}[$H$-eigenvector centrality (HEC)]
Let $\hypergraph$ be a strongly connected $m$-uniform hypergraph with adjacency tensor $\cmT$.
Then the $H$-eigenvector centrality vector for $\hypergraph$ is the positive real vector $\vc$ 
satisfying $\cmT\vc^{m-1} = \lambda \vc^{m}$ and $\|\vc\|_1 = 1$ for some eigenvalue $\lambda > 0$.
\end{definition}

Computing the HEC vector is considerably easier than computing a ZEC vector.
Simple power-method-like algorithms are guaranteed to converge and work
well in practice~\cite{Liu-2010-algorithm,Ng-2010-algorithm,Zhou-2013-algorithm,Gautier-2017-Perron,Gautier-2018-unifying}.

\subsection{Analysis of an illustrative example: the sunflower with singleton core}

\begin{figure}
\centering
\begin{minipage}[c]{0.22\columnwidth}
\includegraphics[width=0.9\columnwidth]{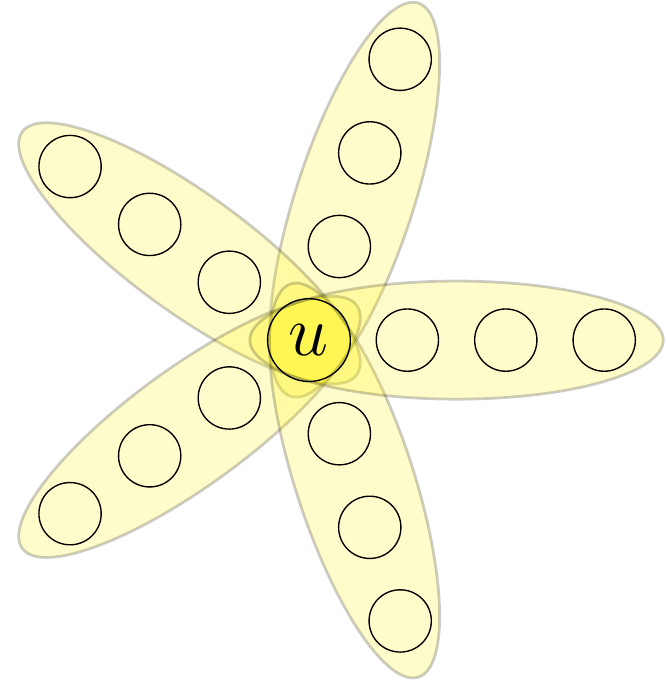}
\end{minipage}
\hfill
\begin{minipage}[c]{0.7\columnwidth}
\begin{tabular}{c c c}
\toprule
centrality & $c_u / c_v$ & $\lim_{m \to \infty} c_u / c_v$ \\ \midrule
CEC & $\frac{2r(m-1)}{\sqrt{m^2 + 4(m - 1)(r - 1)} + m - 2}$ & $2r$ \\
ZEC  & $r^{1/2}$\phantom{!}\raisebox{1ex}{$\ast$}
& $r^{1/2}$ \\
HEC & $r^{1/m}$\phantom{!} & $1$ \\
\bottomrule
\end{tabular}
\end{minipage}
\vspace{-3mm}
\caption{%
Hypergraph eigenvector centralities on sunflowers with singleton cores.
{\bf (Left)} A $4$-uniform, $5$-petal sunflower hypergraph with core $C = \{u\}$. 
Each petal is a hyperedge, marked here with a yellow-shaded ellipse around the nodes.
{\bf (Right)} Finite and asymptotic ratios of the centrality of the center node $u$ to
any other node $v$ (which all have the same centrality) in a general $m$-uniform, 
$r$-petal sunflower hypergraph with singleton core.
The ZEC scores have no dependence on $m$, and the HEC scores
tend to uniform as $m$ grows. \\
\raisebox{1ex}{$\ast$}This ratio holds when $m > 3$ and \emph{can}
hold when $m = 3$ (see \cref{prop:ZEC_sunflower}).
}
\label{fig:sunflower}
\end{figure}

A \emph{sunflower} hypergraph has a hyperedge set $E$ with a common pairwise intersection. 
Formally, for any hyperedges (called \emph{petals}) $A,B \in E$, $A \cap B = C$. The common
intersection $C$ is called the \emph{core}.
The sunflower is similar to the star graph, which has been used to
evaluate centralities in social networks~\cite{Ruhnau-2000-evec}.
Here, we use sunflowers as an illustrative example for the behavior of our 
three hypergraph eigenvector centralities.
We specifically consider sunflowers with $r$ petals where the core
is a singleton, i.e., $A \cap B = \{u\}$ for any $A, B \in E$ (\cref{fig:sunflower}, left).
Below, we derive analytic solutions for the centrality of each method (see also \cref{fig:sunflower}, right).
In all cases, the hypergraph centralities ``do the right thing," namely the center
node $u$ has the largest centrality. 
However, the behavior of the three centralities differ.

\paragraph*{CEC}
Let $v$ be some other node than $u$.
We assume that the centrality $c_v$ is equal to some constant $z$ for all nodes $v$ and show
that we get a positive eigenvector. The Perron-Frobenius theorem then gives us uniqueness.
Recall that $r$ is the number of petals in the hypergraph.
Under these assumptions, the CEC equations satisfy
\begin{align*}
\textstyle \lambda c_u 
&= \textstyle \sum_{(u, v_1, \ldots, v_{m-1}) \in E} \sum_{j=1}^{m-1} c_{v_j} \implies \lambda c_u = r(m-1)z \\
\lambda c_v &= c_u + (m - 2)z
\implies \lambda^2z = r(m - 1)z + \lambda(m - 2)z \\
&\implies \lambda = \textstyle \frac{1}{2}\sqrt{m^2 + 4(m-1)(r-1)} + m - 2 
\implies c_u = \textstyle \frac{2r(m-1)}{\sqrt{m^2 + 4(m - 1)(r - 1)} + m - 2}c_v \nonumber
\end{align*}
Since $m \ge 2$ and $c_v > 0$, $c_u$ and $\lambda$ are both positive for any positive value $c_v$.
Some algebra shows that $c_u > c_v$ for finite $m$ and $r \ge 2$:
\begin{align*}
\textstyle 
\frac{c_u}{c_v} = \frac{2r(m-1)}{\sqrt{m^2 + 4(m - 1)(r - 1)} + m - 2} 
= \frac{r}{\sqrt{\frac{1}{4}\left(\frac{m}{m-1}\right)^2 + \frac{r - 1}{m-1}} + \frac{1}{2} \cdot \frac{m - 2}{m-1}} 
> \frac{r}{\sqrt{1 + \frac{r - 1}{m-1}} + \frac{1}{2}} 
\ge
\begin{cases}
\frac{r}{\sqrt{2} + \frac{1}{2}} > 1 & m \ge r \\
\frac{r}{\sqrt{r} + \frac{1}{2}} > 1  &  m < r
\end{cases}.
\end{align*}
Lastly, we only need to choose $c_v > 0$ and normalize so that $\| \vc \|_1 = 1$.

\paragraph*{ZEC}
The Perron-Frobenius theorem for tensor $Z$-eigenvectors does not preclude the existence
of multiple positive eigenvectors with positive eigenvalues.
We indeed see the non-uniqueness for the sunflower, but only in the $3$-uniform case.
We first show the following lemma, which states that the centrality of the non-core nodes
in any petal must be the same.

\begin{lemma}\label{lem:petal}
In any sunflower whose common intersection is a singleton $\{u\}$, the ZEC scores
of all nodes in the same petal---except for $u$---are the same.
\end{lemma}
\begin{proof}
Let $c_u$ be the centrality score of node $u$.
Let $w$ be any other node in an arbitrary petal $P$, whose centrality score is $c_w$.
The ZEC equations satisfy
\begin{equation}\label{eq:zec_uniq1}
\textstyle \lambda c_w = c_u \prod_{v \in P \backslash \{u, w\}} c_v \implies \lambda c_w^2 = c_u \prod_{v \in P \backslash \{u\}}c_v = \lambda c_{w'}^2 \text{ for any $w' \in P$}.
\end{equation}
\end{proof}
We next characterize exactly when the sunflower has a unique ZEC vector.
\begin{proposition}\label{prop:ZEC_sunflower}
Let $\hypergraph$ be an $m$-uniform sunflower with singleton core $\{u\}$ and petals $\{P\}$.
\begin{enumerate}
\item If $m \neq 3$, the unique $Z$-eigenvector centrality score $\vc$ for $\hypergraph$
is given by $c_u / c_v = \sqrt{r}$, where $v$ is any node other than $u$ and $c_v$ is a constant over nodes $v \neq u$.
\item If $m = 3$, there are infinite $Z$-eigenvector centrality scores for $\hypergraph$;
any vector with $c_v = c_P$ for $v \in P, v \neq u$, and $c_u = \sqrt{\sum_{P}c_P^2}$ are $Z$-eigenvector centrality scores.
\end{enumerate}
\end{proposition}
\begin{proof}
By \cref{lem:petal}, each node other than $u$ has centrality $c_P$, where $P$ is
the petal to which the node belongs.
Re-writing \cref{eq:zec_uniq1} in terms of $c_P$ gives
\begin{equation}\label{eq:lambda_ratio_Z}
\textstyle \lambda c_P^2 = c_u \prod_{v \in P \backslash \{u\}}c_P = c_uc_P^{m-1}
\implies \lambda / c_u = c_P^{m-3} 
\end{equation}
This implies that $c_P^{m-3} = c_{P'}^{m-3}$ for any petals $P$ and $P'$. 
Assume $m \neq 3$. Then $c_P = c_{P'}$ since $c_P > 0$.
Let $z$ be the centrality of any node $v \neq u$.
The ZEC equations satisfy
\begin{align}
\lambda c_u &= \textstyle \sum_{P} z^{m-1} = rz^{m-1} \label{eq:zec1} \text{ and } \lambda z = c_uz^{m-2}
\end{align}
Combining these equations gives $\lambda^2 = rz^{2(m-2)}$, or $\lambda = \sqrt{r}z^{m-2}$.
Plugging this expression for $\lambda$ into \cref{eq:zec1} gives
\begin{align}
\textstyle c_u / c_v = rz^{m-1} / (\lambda z) = rz^{m-1}/(\sqrt{r}z^{m-1}) = \sqrt{r}.
\end{align}

Now assume $m = 3$. Then $\lambda = c_u$ by \cref{eq:lambda_ratio_Z}.
Let $c_P > 0$ be an arbitrary constant for each petal $P$ and define $c_u^2 = \sum_{P} c_P^2$.
We now just check that the $Z$-eigenvector equations hold.
\begin{equation}
\textstyle c_u^2 = \lambda c_u = \sum_{P} \prod_{v \in P \backslash \{u\}} c_v = \sum_{P} c_P^2,
\end{equation}
which holds by the definition of $c_u^2$. Our choice of $c_P$ was any positive real number, and for any node $w \neq u$ in petal $P = \{u, v, w\}$, the ZEC equation is 
$c_uc_P = c_uc_v = \lambda c_w = c_uc_w = c_uc_P$.
\end{proof}

Surprisingly, when $m = 3$, the non-center nodes can have different ZEC scores,
even though the symmetry of the problem would suggest that they would be the same.
Also surprisingly, all scores are independent of $m$, the number of nodes in a hyperedge.
However, ZEC is consistent in the sense that the center node always has the largest centrality score.

\paragraph*{$H$-eigenvector centrality}
\Cref{thm:PFH} gives us uniqueness of a positive vector with positive eigenvalue.
We again assume that $c_v = z$ for any node $v \neq u$.
The HEC equations satisfy
\begin{align}
\textstyle \lambda c_u^{m-1} &= \textstyle  \sum_{(u, v_1, \ldots, v_{m-1})} \prod_{j=1}^{m-1} c_{v_j} 
\implies \lambda c_u^{m-1} = rz^{m-1} \label{eq:hec1} \\
\lambda c_v^{m-1} &= c_uz^{m-2} \implies \lambda z^{m-1} = c_uz^{m-2} 
\implies c_u = \lambda z \label{eq:hec2}.
\end{align}
Plugging in $c_u = \lambda z$ into \cref{eq:hec1} gives $\lambda (\lambda z)^{m-1} = rz^{m-1} \implies \lambda = r^{1/m}$. Thus, $c_u / c_v = r^{1/m}$ for $v \neq u$, and $c_u / c_v \to 1$ if the number of petals $r$ is fixed 
and the uniformity $m$ grows large.

\subsection{Recap: which centrality should we use? \nopunct}\label{sec:recap}
We derived three hypergraph eigenvector centralities.
The appeal of CEC is that we only need to rely on the familiar, i.e., we can just use nonnegative matrix theory.
However, CEC does not incorporate any interesting nonlinear structure, whereas ZEC and HEC incorporate nonlinearity.
HEC is certainly attractive computationally---simple algorithms can compute a unique eigenvector centrality vector.
We don't have scalable algorithms guaranteed to compute a ZEC vector, and even worse, the ZEC vector may not be unique. 
Moreover, the non-uniqueness can show up on simple hypergraphs, as we saw with the sunflower.
Both CEC and HEC have a proper dimensional analysis, while ZEC does not.
On the other hand, ZEC can carry the same rank-1 approximation interpretation as standard graph eigenvector centrality.
Also, in the asymptotics of the sunflower analysis (\cref{fig:sunflower}, right), the HEC score of the center node approaches that of the other nodes, 
while the relative CEC and ZEC scores of the center node to the others are constants that only depend on the number of hyperedges.

So which centrality should we use? 
Our analysis suggests that none is superior to all others.
As is the case with graph centralities in general, the scores are not useful in a vacuum.
Instead, we can use various centralities to study data.
For example, multiple centralities provide more features that can be used for machine learning tasks.
In the next section, we show that the three hypergraph centralities can provide qualitatively different results on real-world data.

\section{Computational experiments and data analysis}\label{sec:data}


\newsavebox{\cmTbox}
\savebox{\cmTbox}{$\cmT$}
\begin{table}[tb]
\setlength{\tabcolsep}{4pt}
\centering
\caption{Summary statistics of datasets. The number of nodes is the dimension of the cubic
adjacency tensor \usebox{\cmTbox} of the largest component of the hypergraph, 
and nnz(\usebox{\cmTbox}) is the number of non-zeros in 
\usebox{\cmTbox}, which we divide by the number of symmetries in the symmetric tensor.}
\vspace{-3mm}
\begin{tabular}{r @{\quad} c c @{\quad} c c @{\quad} c c}
\toprule
& \multicolumn{2}{c}{\!\!\!\!\!\!\!\!\!\!\!3-uniform}
& \multicolumn{2}{l}{\phantom{!!!!!!}4-uniform}
& \multicolumn{2}{l}{\phantom{!!!!!!}5-uniform} \\
\cmidrule(l{-2pt}r{10pt}){2-3}
\cmidrule(l{-2pt}r{10pt}){4-5} 
\cmidrule(l{-2pt}r{0pt}){6-7}
dataset & \# nodes & $\frac{\text{nnz}(\cmT)}{6}$ & \# nodes & $\frac{\text{nnz}(\cmT)}{24}$ & \# nodes & $\frac{\text{nnz}(\cmT)}{120}$ \\
\midrule
$\ngrams$ & 30,885 & 888,411 & 23,713 & 957,904 & 24,996 & 995,952 \\
$\tags$       & 2,981 & 279,369 & 2,856 & 145,676 & 2,564 & 25,475\\
$\dawn$     & 1,677 & 41,225 & 1,447 & 29,829 & 1,212 & 15,690 \\
\bottomrule
\end{tabular}
\label{tab:summary-stats}
\end{table}

We now analyze our proposed eigenvector centralities on three real-world datasets.
We construct a 3-uniform, 4-uniform, and 5-uniform hypergraph from each of the three datasets 
(summary statistics are in \cref{tab:summary-stats}),
so there are 9 total hypergraphs for our analysis.
For each of the 9 hypergraphs, we computed the CEC, ZEC, and HEC scores
on the largest connected component of the hypergraph.
We used Julia's \texttt{eigs} routine to compute the CEC scores,
the dynamical systems algorithm by \citet{Benson-2018-Z} to compute the ZEC scores,
and the NQI algorithm~\cite{Ng-2010-algorithm} to compute the HEC scores.
The software used to compute the results in this section is available at
\url{https://github.com/arbenson/Hyper-Evec-Centrality}.

As discussed above, the ZEC vector need not be unique.
We computed 100 ZEC vectors using random starting points and found
that, for some datasets, the algorithm always converges to the same eigenvector
and in others, the algorithm converges to a few different ones. For the purposes
of our analysis, we use the ZEC vector to which convergence was most common.
However, any of the ZEC vectors is a valid centrality.
(One could also possibly take the mean of several ZEC vectors, 
but linear combinations of $Z$-eigenvectors
are not necessarily $Z$-eigenvectors, unlike the matrix case.)

\paragraph*{$\ngrams$}
These hypergraphs are constructed from the most frequent $N$-grams in the 
Corpus of Contemporary American English (COCA)~\cite{Davies-2011-ngrams}.%
\footnote{\url{https://www.ngrams.info}}
An $N$-gram is a sequence of $N$ words (or parts of words, but we will just say ``words'') that appear contiguously in text.
Here, we use the million most frequent $N$-grams dataset from COCA for $N = 3, 4, 5$ to compose hyperedges.
We construct $m$-uniform hypergraphs ($m = 3, 4, 5$) as follows. The set of nodes in the $m$-uniform hypergraph
correspond to all words appearing in at least one of the $m$-grams in the corpus.
There is a hyperedge between $m$ nodes if the corresponding $m$ words (appearing in any permutation order)
make up one of the $m$-grams appearing in the corpus.
For each hypergraph, we analyze its largest connected component, which is given by taking the node set $L$
from the largest connected component of the graph discussed in \cref{def:conn}, and only keeping
the hyperedges comprised entirely of nodes in $L$.

\Cref{tab:ngrams-ranks} lists the top 20 ranked words according to the CEC, ZEC, and HEC
scores for each of the three hypergraphs.
Many of the top-ranked words are so-called stop words, such as ``the,'' ``and,'' and ``to''; furthermore,
nearly all of the top 20 ranked words for CEC and HEC are stop words or conjunctions,
regardless of the size of the $N$-gram.
This is perhaps not surprising, given that stop words are by definition common in natural language
(stop words also form important clusters in tensor-based clustering of $N$-gram data~\cite{Wu-2016-cocluster}).
The same is true of the ZEC scores, but only for the 3-grams and 4-grams.
In the 5-uniform hypergraph,
the word ``world'' has rank 12 with ZEC, rank 64 with CEC, and rank 84 with HEC;
and the word ``people'' has rank 14 with ZEC, rank 39 with CEC, and rank 44 with HEC.

To better quantify the relationship between the centralities, we computed the Spearman's
rank correlation coefficient between components of each centrality vector. Specifically,
for each of CEC, ZEC, and HEC, we find the top $k$ ranked nodes and compute the 
rank correlation on the sub-vectors consisting of these nodes with the other two centrality vectors.
For example, if $k = 100$, we compute the top 100 nodes according to the CEC vector, take
the length-100 sub-vector corresponding to the same nodes in the ZEC vector, and compute the rank
correlation between the vectors. This is repeated for all six possible pairs of vectors and
plotted as a function of $k$ (\cref{fig:ngrams-corrs}).

As a function of $k$, the rank correlations in this dataset tend to have local minima for $k$ between 20 and a few hundred.
Larger values of $k$ catch the tail of the distribution for which there is less difference in ranking,
which leads to the increase in the correlation for large $k$.
We also see that the correlations tend to decrease as we increase the uniformity of the hypergraph.
In other words, the three centrality measures become more different when considering larger multi-way relationships.
Finally, the rank correlations reveal that ZEC ranks the top nodes (beyond the top 20) substantially differently than CEC and HEC.

\paragraph*{$\tags$}
Ask Ubuntu\footnote{\url{https://askubuntu.com}} is a Stack Exchange forum, 
where Ubuntu users and developers ask, answer, and discuss questions.
Each question may be annotated with up to five tags to aid in classification.
We construct 3-uniform, 4-uniform, and 5-uniform hypergraphs from a previously collected dataset of tag co-appearances
in questions~\cite{Benson-2018-simplicial}.
Specifically, the nodes of the hypergraphs represent tags. We add each possible hyperedge to the $m$-uniform hypergraph
if the corresponding $m$ tags were all simultaneously used to annotate at least one question on the web site (the question
could also have contained other tags; for constructing the hyperedge, we only care if the $m$ tags were used
for the question).
Finally, as before, we use the largest component of the hypergraph.

\Cref{tab:tags-ranks} lists the top 10 nodes ranked by CEC, ZEC, HEC for each of the three hypergraphs.
With the 3-uniform hypergraph, these top-ranked nodes are roughly the same for each centrality measure,
with major Ubuntu version numbers (``12.04,'' ``14.04,'' and ``16.04'') near or at the top of each list.
When moving to 4-uniform and 5-uniform hypergraphs, the version numbers remain highly ranked, but not
the most highly ranked. 
ZEC finds tags related to the Windows operating system relatively more important.
For example, the tags 
``windows-8'', ``windows'', and ``windows-7'' are
ranked 8, 9, and 10 with ZEC for the 5-uniform hypergraph
but ranked 28, 22, and 26 with CEC and 21, 18, and 20 with HEC. 
Furthermore, ZEC ranks ``windows,'' ``windows-xp,'' ``windows-vista,'' ``windows-7,'' ``windows-8,''  and ``windows-10'' 
higher than CEC and HEC for all three hypergraphs. 
We conclude that ZEC provides complimentary information to the centralities for this dataset.

\Cref{fig:tags-corrs} lists the same rank correlations as described above.
We again see that all centrality vectors are relatively correlated for the 3-uniform hypergraph but
less so as we increase the order of the hypergraph.
The sub-vector corresponding to the top 10 ranked CEC nodes has only 0.05 rank correlation with the same ZEC sub-vector
for the 4-uniform hypergraph.

\paragraph*{$\dawn$}
The Drug Abuse Warning Network (DAWN) is a national health surveillance system in hospitals throughout the United States
that records the drug use reported by patients visiting emergency rooms. Here, drugs include illicit substances, prescription medication, 
over-the-counter medication, and dietary supplements.
We use a dataset that aggregates 8 years of DAWN reports~\cite{Benson-2018-simplicial} to construct $m$-uniform hypergraphs for $m = 3,4,5$.
The nodes in each hypergraph correspond to drugs.
We add a hyperedge on $m$ nodes if there is at least one patient that reports using exactly that combination of $m$ drugs.
Again, we use the largest component of the hypergraph.

We again list the top 10 ranked nodes by the three centrality vectors for each of the three hypergraphs (\cref{tab:DAWN-ranks})
as well as the same rank correlation statistics (\cref{fig:DAWN-corrs}).
In this dataset, we see near agreement between the three centrality vectors across the 4-uniform and 5-uniform hypergraphs.
For example, the rank correlations remain above 0.75 for the entire 4-uniform hypergraph for all measured top $k$ sub-vectors.
Alcohol is consistently ranked near the top, which is unsurprising given its pervasiveness in emergency department
visits, especially in combination with other drugs~\cite{Crane-2013-DAWN}.

The ranking from the ZEC vector is substantially different from HEC and CEC for the 3-uniform hypergraph.
Leading sub-vectors of ZEC actually have negative rank correlation with the corresponding HEC and CEC sub-vectors.
As with the $\ngrams$ and $\tags$ datasets, we again conclude that ZEC provides complimentary information for the centralities.


\begin{figure}
\centering
\captionof{table}{Top 20 nodes with
largest centralities for CEC, ZEC, and HEC
for the three hypergraphs constructed from the frequent $n$-grams.
Many stop words appear as the top-ranked nodes, but ZEC picks
up on non-stop words such as ``world'' and ``people''
on the hypergraph constructed from frequent $5$-grams.
}
\setlength{\tabcolsep}{4pt}
\begin{tabular}{l @{\quad} l l l @{\qquad} l l l @{\qquad} l l l}
\toprule
& \multicolumn{3}{l}{\!\!3-uniform}
& \multicolumn{3}{l}{\!\!4-uniform}
& \multicolumn{3}{l}{\!\!5-uniform} \\
\cmidrule(l{-4pt}r{18pt}){2-4}
\cmidrule(l{-4pt}r{18pt}){5-7} 
\cmidrule(l{-4pt}r{1pt}){8-10} 
& CEC & ZEC & HEC & CEC & ZEC & HEC & CEC & ZEC & HEC \\ 
\midrule
1 & the & the & the & the & the & the & the & the & the \\
2 & of & to & to & of & of & to & of & of & to \\
3 & in & and & a & to & to & of & to & in & of \\
4 & and & a & and & in & in & a & in & and & a \\
5 & to & that & of & a & and & in & a & to & that \\
6 & a & in & in & and & that & that & and & that & in \\
7 & that & of & that & that & a & and & that & on & and \\
8 & on & is & for & on & is & is & is & is & i \\
9 & for & for & is & is & on & it & on & a & it \\
10 & with & it & on & for & for & was & be & one & n't \\
11 & is & on & was & was & was & i & for & for & is \\
12 & was & was & with & be & it & you & was & world & you \\
13 & from & with & it & with & you & have & it & with & have \\
14 & at & you & as & it & with & be & have & people & was \\
15 & by & this & you & at & one & for & i & end & be \\
16 & as & as & i & have & be & on & with & part & do \\
17 & his & i & this & i & have & he & at & at & he \\
18 & it & have & be & as & all & n't & you & first & for \\
19 & be & at & have & he & this & with & n't & rest & on \\
20 & are & not & at & you & at & not & as & was & we \\
\bottomrule
\end{tabular}
\label{tab:ngrams-ranks}

\medskip

\includegraphics[width=0.32\columnwidth]{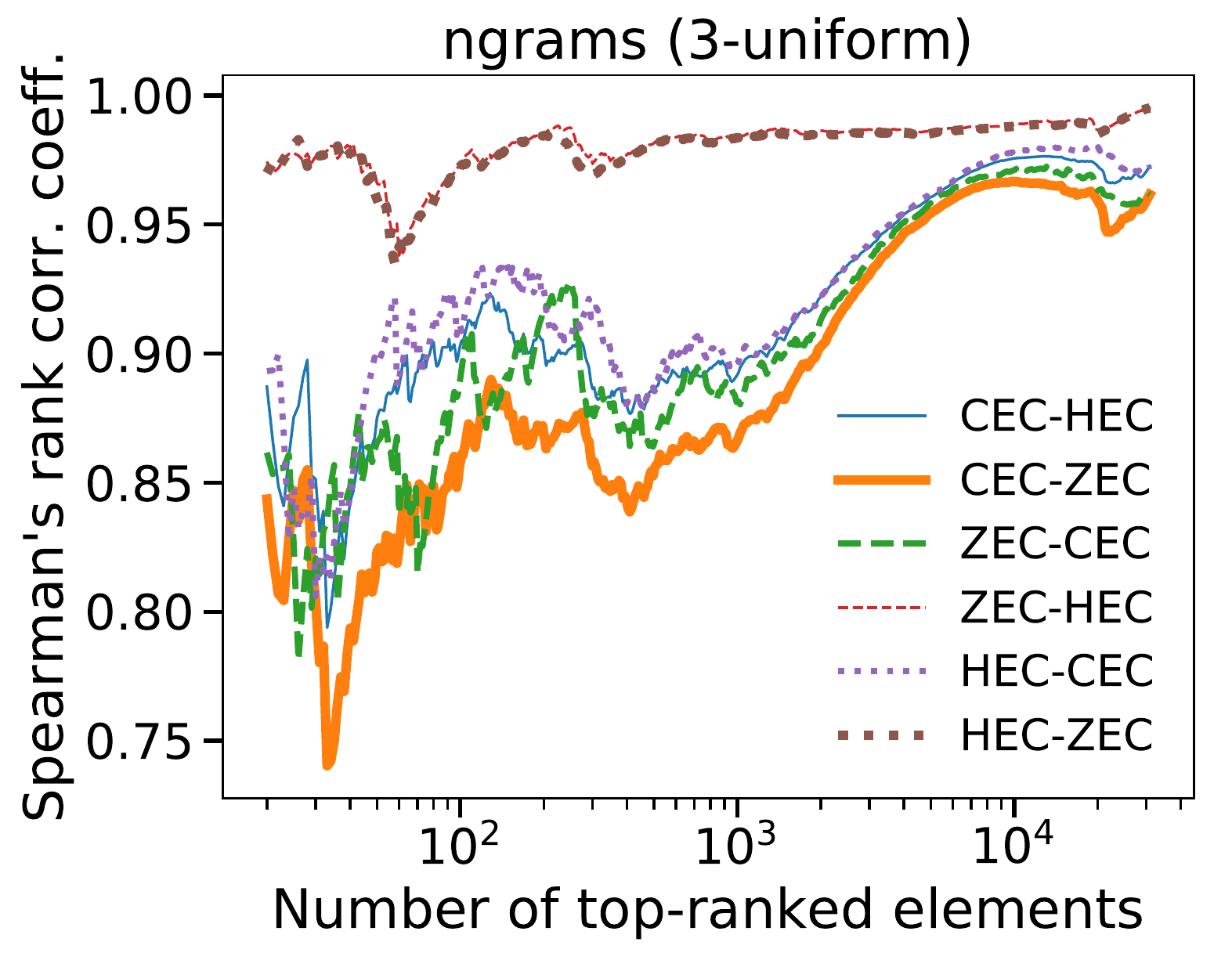}
\includegraphics[width=0.32\columnwidth]{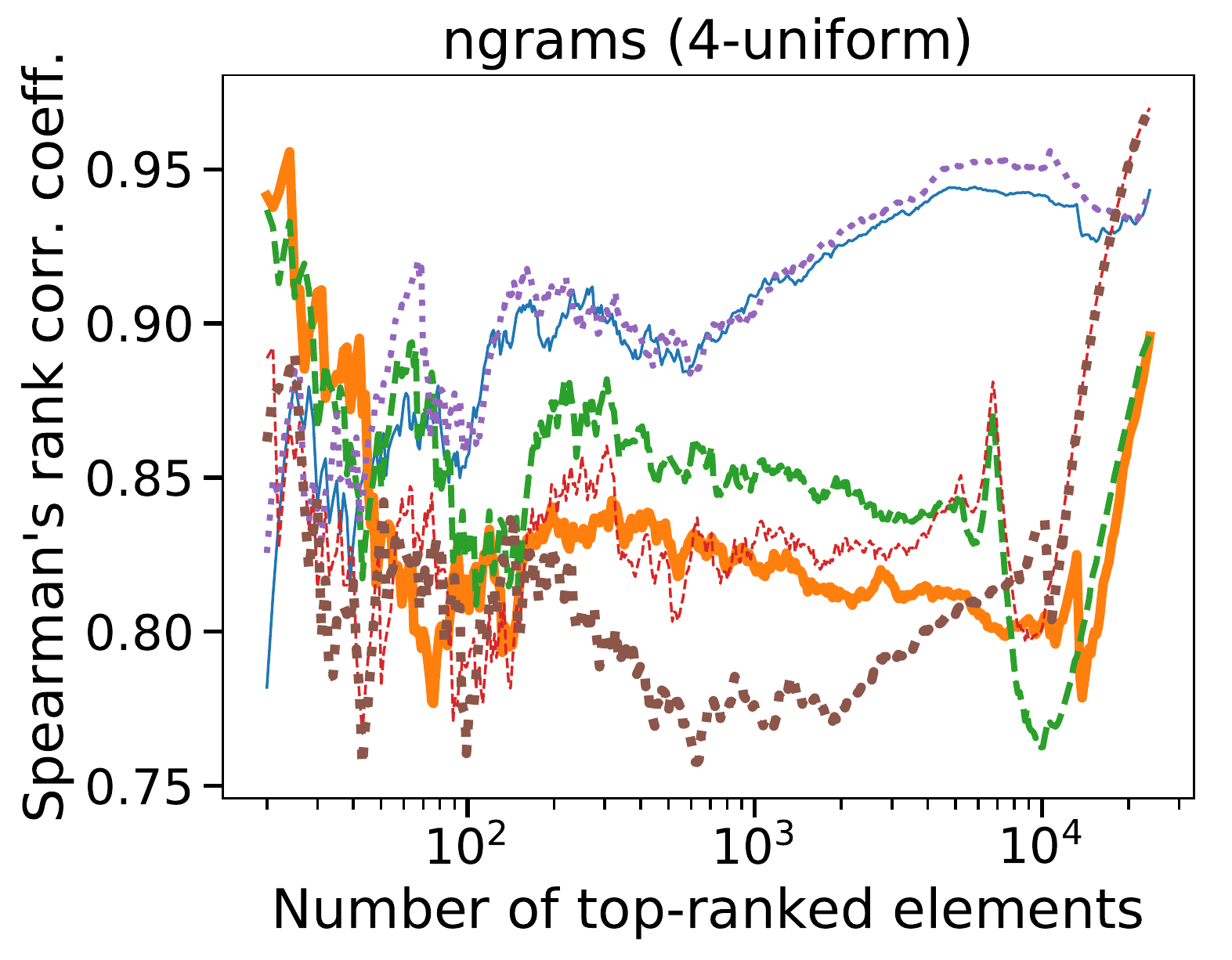}
\includegraphics[width=0.32\columnwidth]{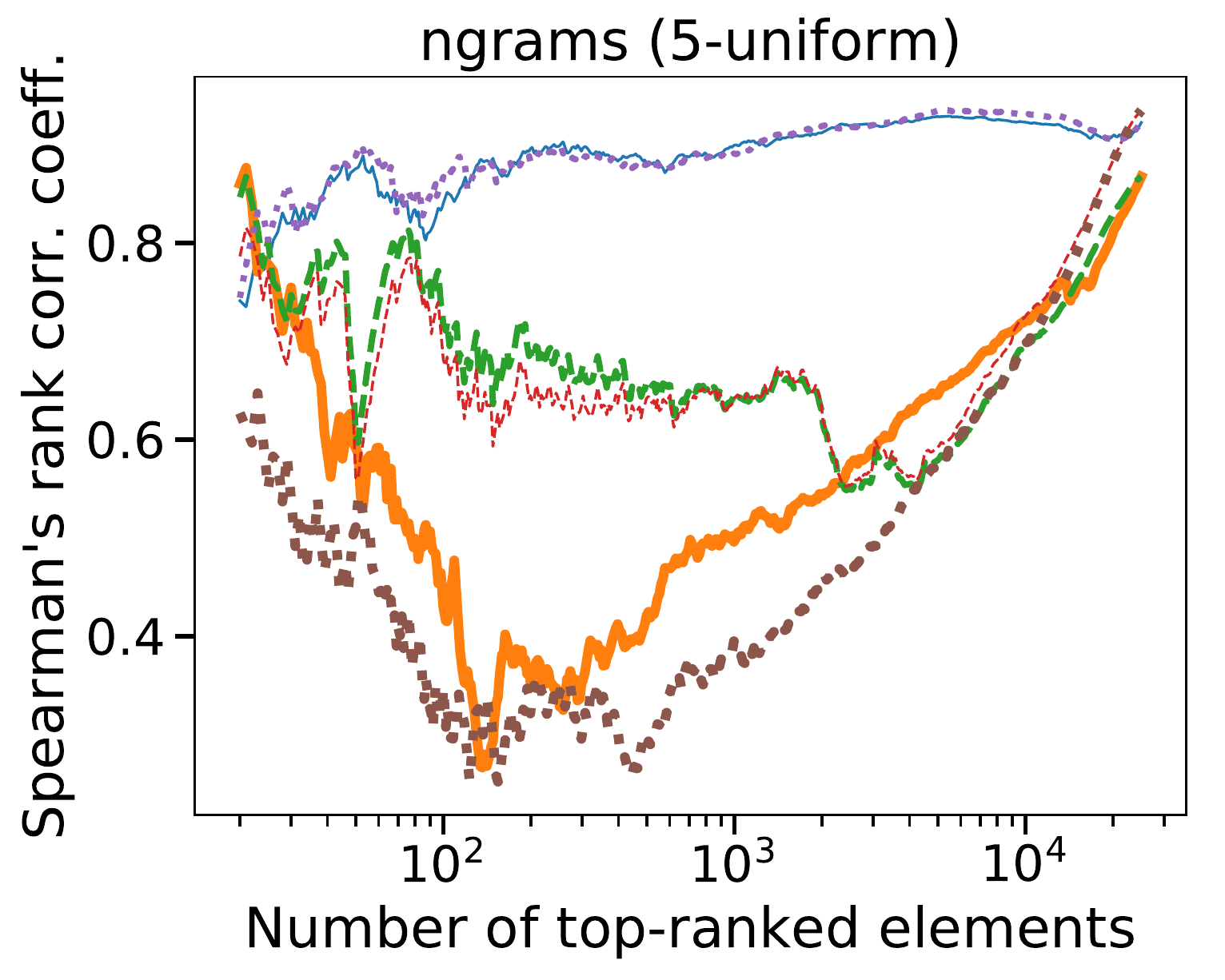}
\captionof{figure}{
Spearman's rank correlation coefficient between the top $k$ ranked
nodes from one centrality measure with the same nodes from the other two centrality
measures on the $N$-grams dataset (the one used to determine the nodes is listed first in the legend).
The rank correlation for the top few hundred nodes with ZEC and the other centralities decreases as the uniformity
of the hypergraph increases, and dips below 0.4 for the 5-uniform hypergraph.
}

\label{fig:ngrams-corrs}
\end{figure} 


\begin{figure}[tb]
\centering
\begin{minipage}{0.55\textwidth}
\centering
\setlength{\tabcolsep}{2pt}
\captionof{table}{Top 10 nodes with
largest centralities for CEC, ZEC, and HEC
for three hypergraphs constructed from the co-tagging
dataset tags-ask-ubuntu.
These highly-ranked nodes are largely the same in the 3-uniform hypergraph.
For the 4-uniform and 5-uniform hypergraphs,
ZEC picks up on Windows-related tags.
Tags related to version numbers are ranked lower as the
uniformity of the hypergraph increases.
}
\scalebox{0.75}{
\begin{tabular}{l @{\quad} l @{\quad} l l l}
\toprule
& & CEC & ZEC & HEC \\
\midrule
\parbox[t]{2mm}{\multirow{10}{*}{\rotatebox[origin=c]{90}{3-uniform}}}
& 1 & 14.04 & 14.04 & 14.04 \\
& 2 & 12.04 & 12.04 & 12.04 \\
& 3 & 16.04 & boot & 16.04 \\
& 4 & server & 16.04 & boot \\
& 5 & command-line & drivers & drivers \\
& 6 & boot & nvidia & command-line \\
& 7 & networking & dual-boot & server \\
& 8 & drivers & server & networking \\
& 9 & unity & command-line & unity \\
& 10 & gnome & upgrade & gnome \\
\midrule
& & CEC & ZEC & HEC \\ \midrule
\parbox[t]{2mm}{\multirow{10}{*}{\rotatebox[origin=c]{90}{4-uniform}}}
& 1 & 14.04 & dual-boot & 14.04 \\
& 2 & boot & boot & boot \\
& 3 & drivers & grub2 & drivers \\
& 4 & 12.04 & partitioning & 12.04 \\
& 5 & 16.04 & uefi & 16.04 \\
& 6 & networking & system-installation & dual-boot \\
& 7 & server & 14.04 & nvidia \\
& 8 & dual-boot & windows & grub2 \\
& 9 & nvidia & installation & networking \\
& 10 & grub2 & 12.04 & partitioning \\
\midrule
& & CEC & ZEC & HEC \\ \midrule
\parbox[t]{2mm}{\multirow{10}{*}{\rotatebox[origin=c]{90}{5-uniform}}}
& 1 & boot & dual-boot & boot \\
& 2 & dual-boot & boot & dual-boot \\
& 3 & 14.04 & grub2 & grub2 \\
& 4 & drivers & partitioning & drivers \\
& 5 & grub2 & uefi & 14.04 \\
& 6 & networking & system-installation & partitioning \\
& 7 & 16.04 & 14.04 & nvidia \\
& 8 & partitioning & windows-8 & 16.04 \\
& 9 & nvidia & windows & 12.04 \\
& 10 & 12.04 & windows-7 & networking \\
\bottomrule
\end{tabular}
}
\label{tab:tags-ranks}
\end{minipage}
\hfill
\begin{minipage}[c]{0.43\columnwidth}
\centering
\includegraphics[width=0.85\textwidth]{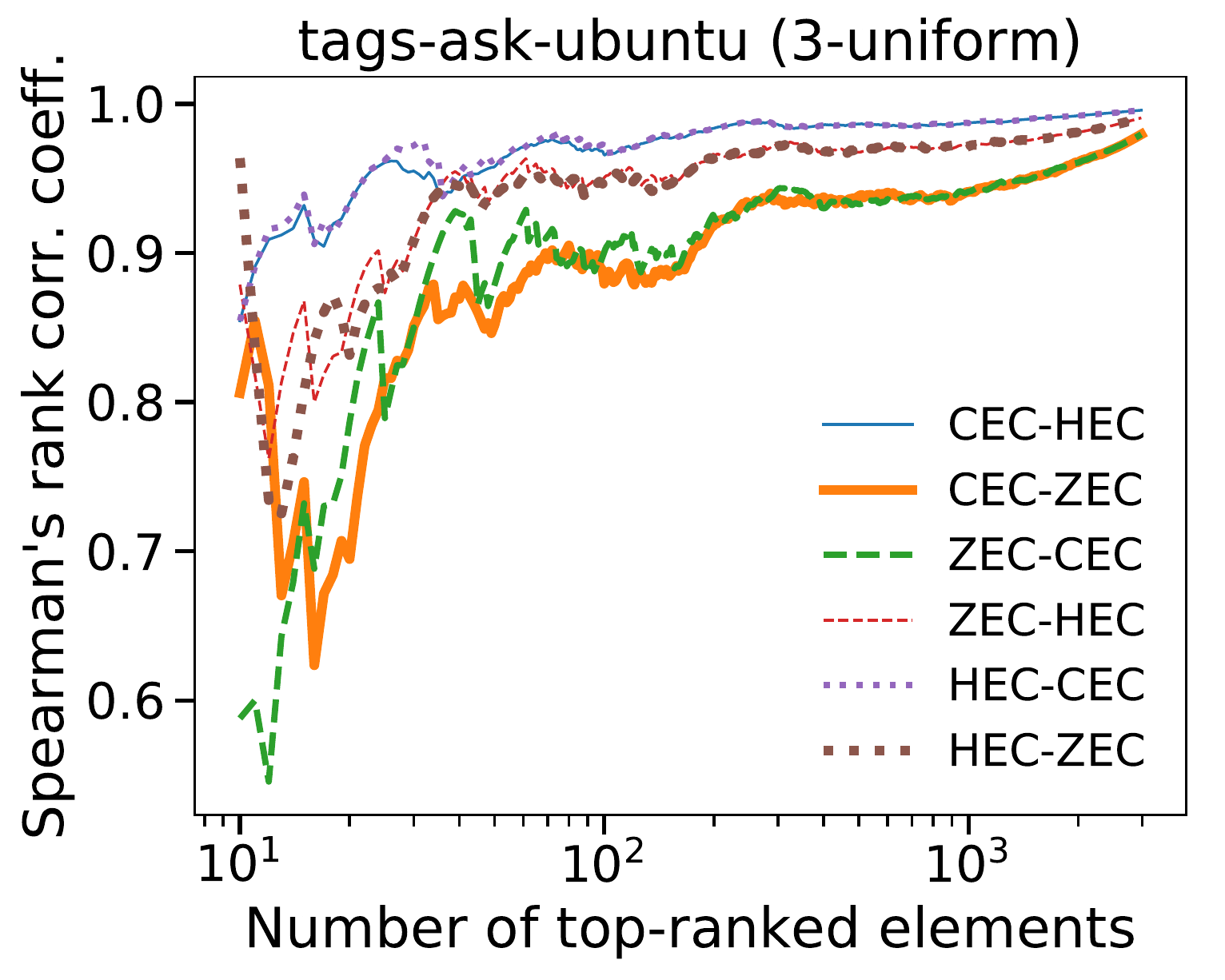}
\includegraphics[width=0.85\textwidth]{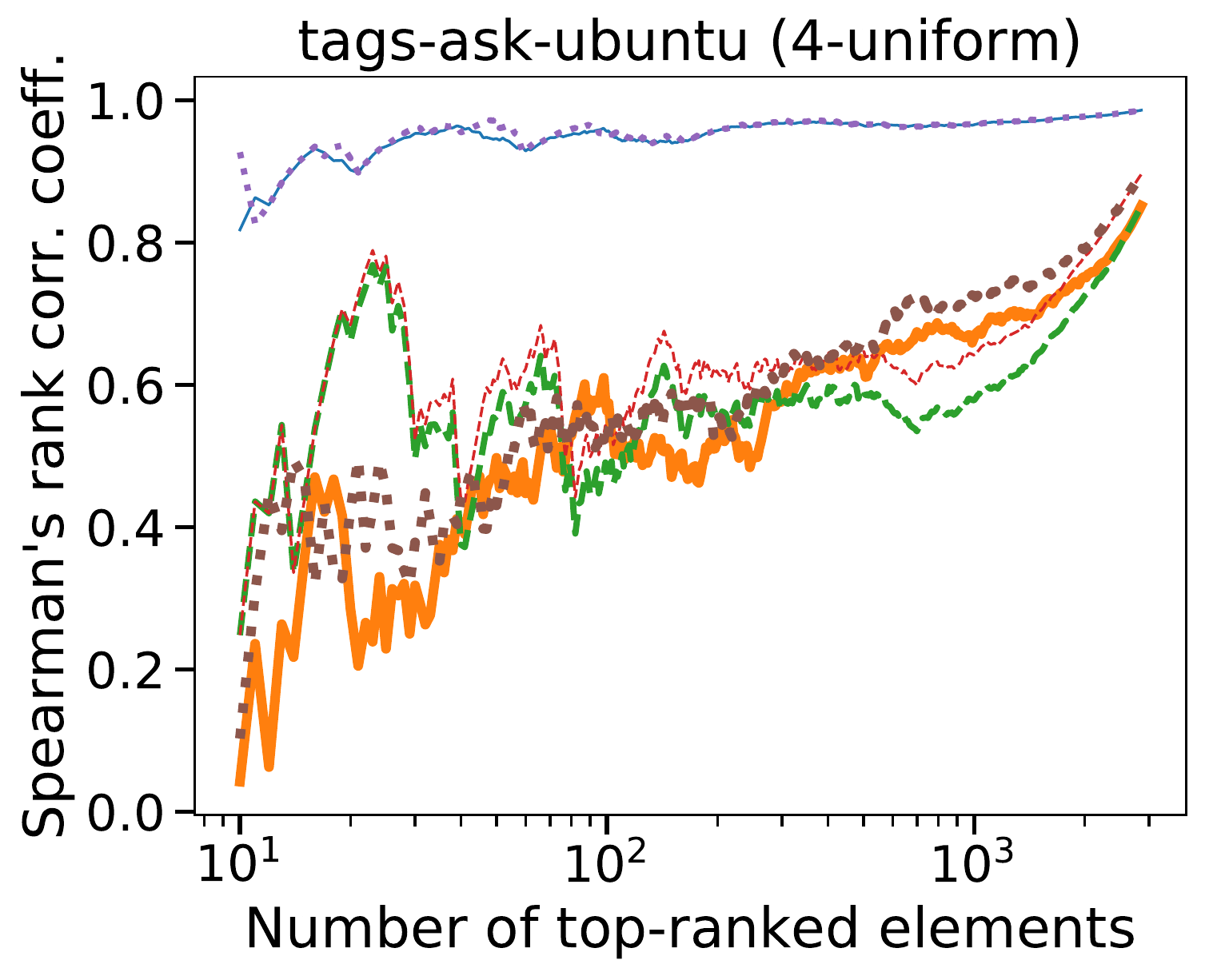}
\includegraphics[width=0.85\textwidth]{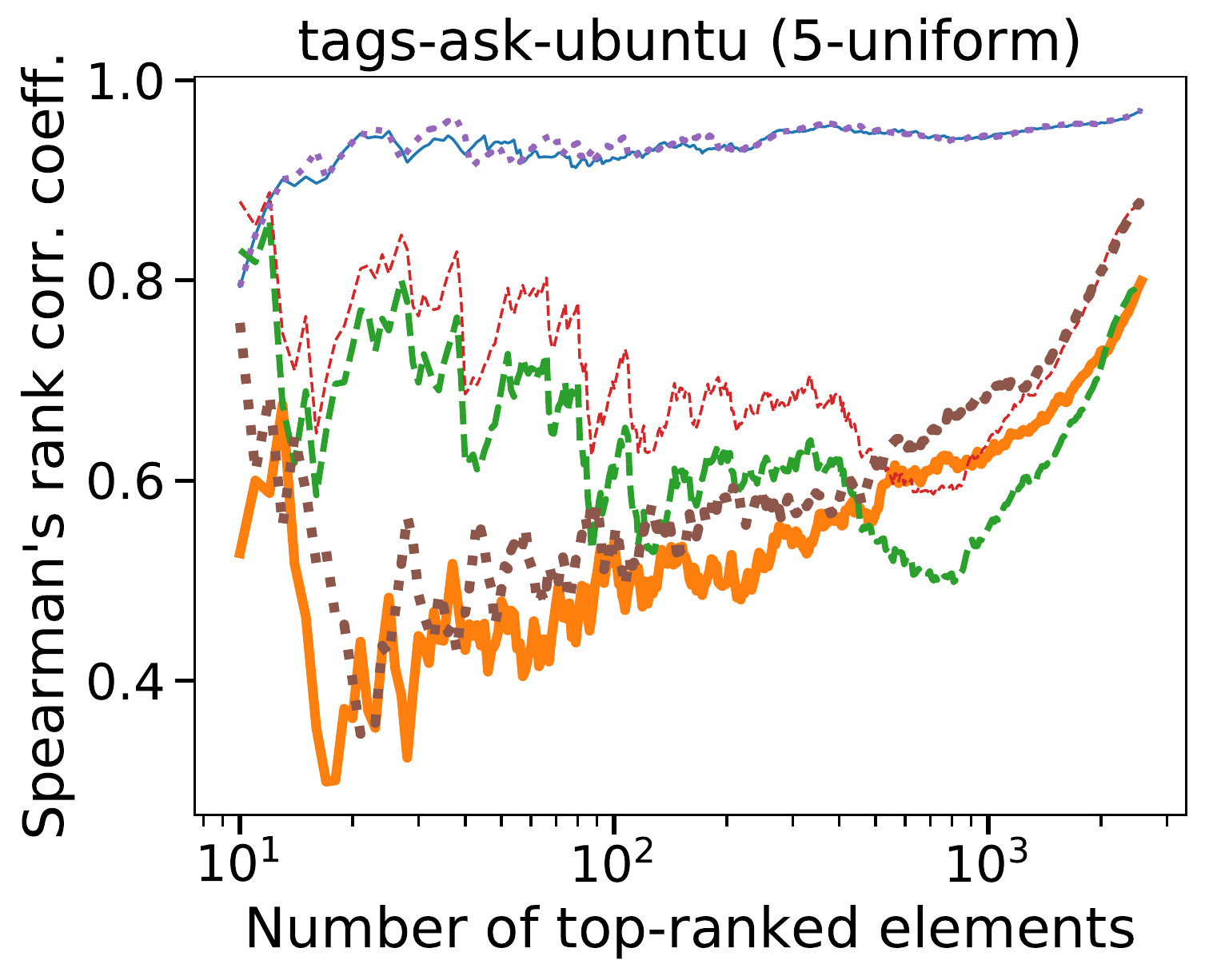}      
\captionof{figure}{Spearman's rank correlation coefficient between the top $k$ ranked
nodes from one centrality measure with the same nodes from the other two centrality
measures on the Ask Ubuntu co-tagging dataset (the one used to determine the nodes is listed first in the legend).
For the 3-uniform hypergraph, all centralities are relatively correlated for $k \ge 30$.
For the 4-uniform and 5-uniform hypergraphs, ZEC tends to be less correlated with
CEC and HEC, which is also seen in the rankings of the top 10 nodes (\cref{tab:tags-ranks}, left).
}
\label{fig:tags-corrs}
\end{minipage}
\end{figure}

\begin{figure}[tb]
\centering
\begin{minipage}{0.55\textwidth}
\centering
\setlength{\tabcolsep}{2pt}
\captionof{table}{Top 10 nodes with
largest centralities for CEC, ZEC, and HEC
for three hypergraphs constructed from the reported sets of drugs
used by patients in emergency room visits in the DAWN dataset.
The highly-ranked nodes in the 4-uniform and 5-uniform
hypergraphs are largely the same and are consistent across the
centrality measures.
}
\scalebox{0.75}{
\begin{tabular}{l @{\quad} l @{\quad} l l l}
\toprule
& & CEC & ZEC & HEC \\
\midrule
\parbox[t]{2mm}{\multirow{10}{*}{\rotatebox[origin=c]{90}{3-uniform}}}
 & 1  & alcohol           & cephalothin     & alcohol           \\
 & 2  & cocaine           & naloxone        & alprazolam        \\
 & 3  & marijuana         & meclizine       & acet.-hydrocodone \\
 & 4  & acet.-hydrocodone & cyclosporine    & clonazepam        \\
 & 5  & alprazolam        & desipramine     & cocaine           \\
 & 6  & clonazepam        & donnatal elixir & marijuana         \\
 & 7  & ibuprofen         & pyridostigmine  & quetiapine        \\
 & 8  & quetiapine        & amoxapine       & lorazepam         \\
 & 9  & acetaminophen     & aspirin         & ibuprofen         \\
 & 10 & lorazepam         & bicalutamide    & zolpidem          \\
\midrule
 & 1  & alcohol           & alcohol           & alcohol           \\
 & 2  & cocaine           & cocaine           & cocaine           \\
 & 3  & marijuana         & marijuana         & marijuana         \\
 & 4  & alprazolam        & alprazolam        & alprazolam        \\
 & 5  & acet.-hydrocodone & acet.-hydrocodone & acet.-hydrocodone \\
 & 6  & clonazepam        & clonazepam        & clonazepam        \\
 & 7  & quetiapine        & heroin            & quetiapine        \\
 & 8  & heroin            & oxycodone         & oxycodone         \\
 & 9  & oxycodone         & methadone         & heroin            \\
 & 10 & lorazepam         & acet.-oxycodone   & acet.-oxycodone   \\
\midrule
 & 1  & alcohol             & cocaine             & alcohol             \\
 & 2  & cocaine             & alcohol             & cocaine             \\
 & 3  & marijuana           & marijuana           & marijuana           \\
 & 4  & alprazolam          & heroin              & alprazolam          \\
 & 5  & acet.-hydrocodone   & alprazolam          & acet.-hydrocodone   \\
 & 6  & clonazepam          & benzodiazepines     & heroin              \\
 & 7  & heroin              & oxycodone           & clonazepam          \\
 & 8  & benzodiazepines     & acet.-hydrocodone   & benzodiazepines     \\
 & 9  & oxycodone           & methadone           & oxycodone           \\
 & 10 & narcotic analgesics & narcotic analgesics & narcotic analgesics \\
\bottomrule
\end{tabular}
}
\label{tab:DAWN-ranks}
\end{minipage}
\hfill
\begin{minipage}[c]{0.43\columnwidth}
\centering
\includegraphics[width=0.85\textwidth]{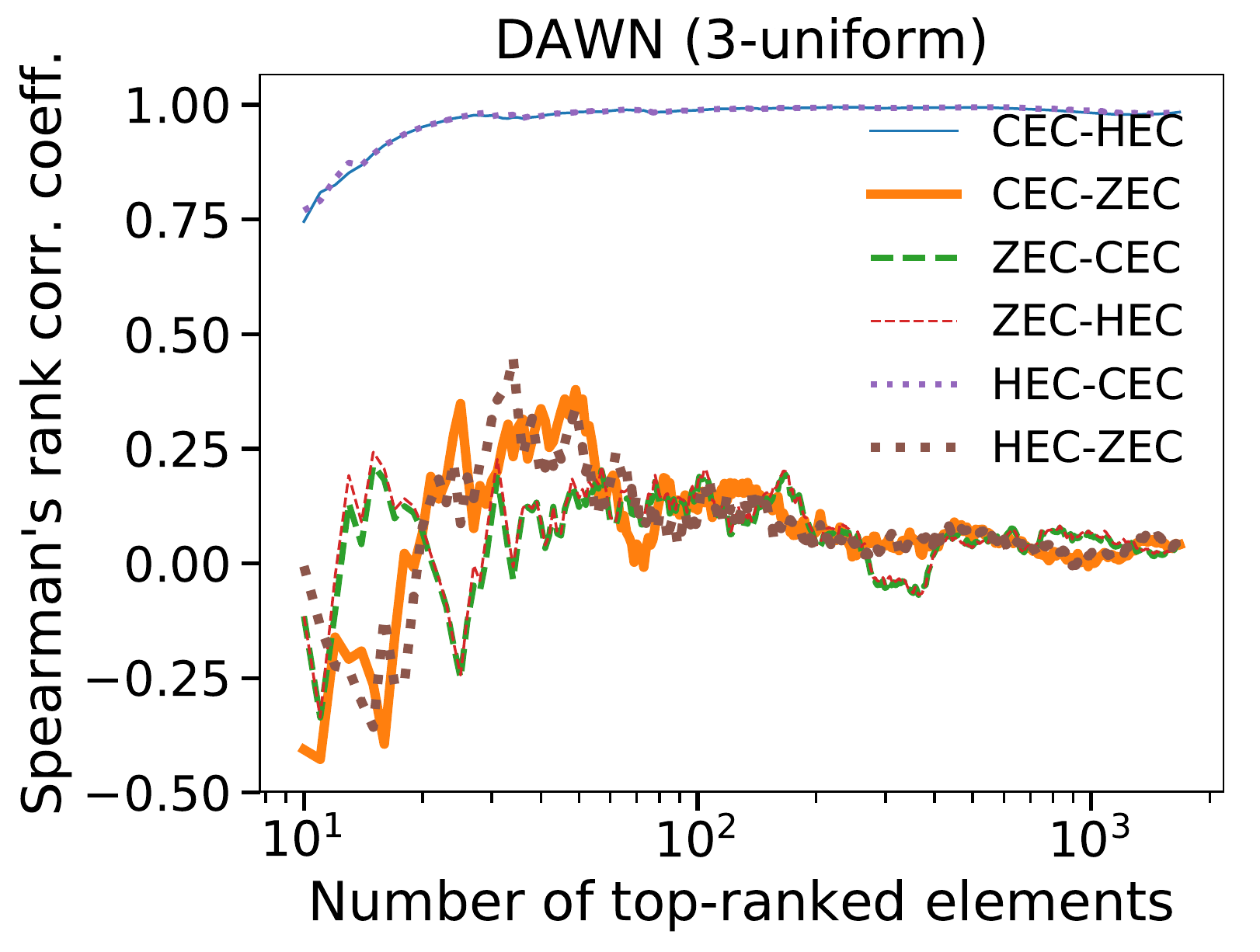}
\includegraphics[width=0.85\textwidth]{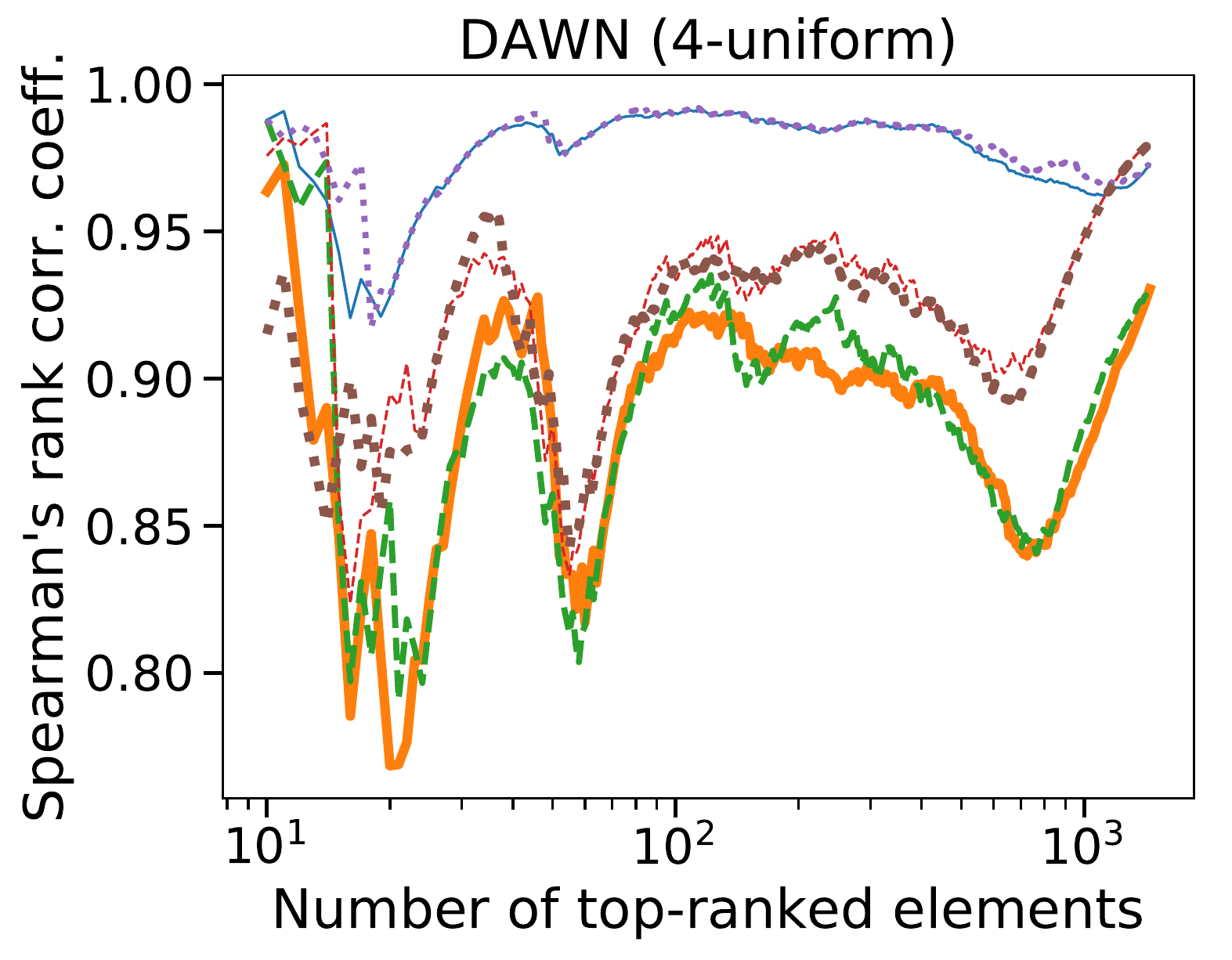}
\includegraphics[width=0.85\textwidth]{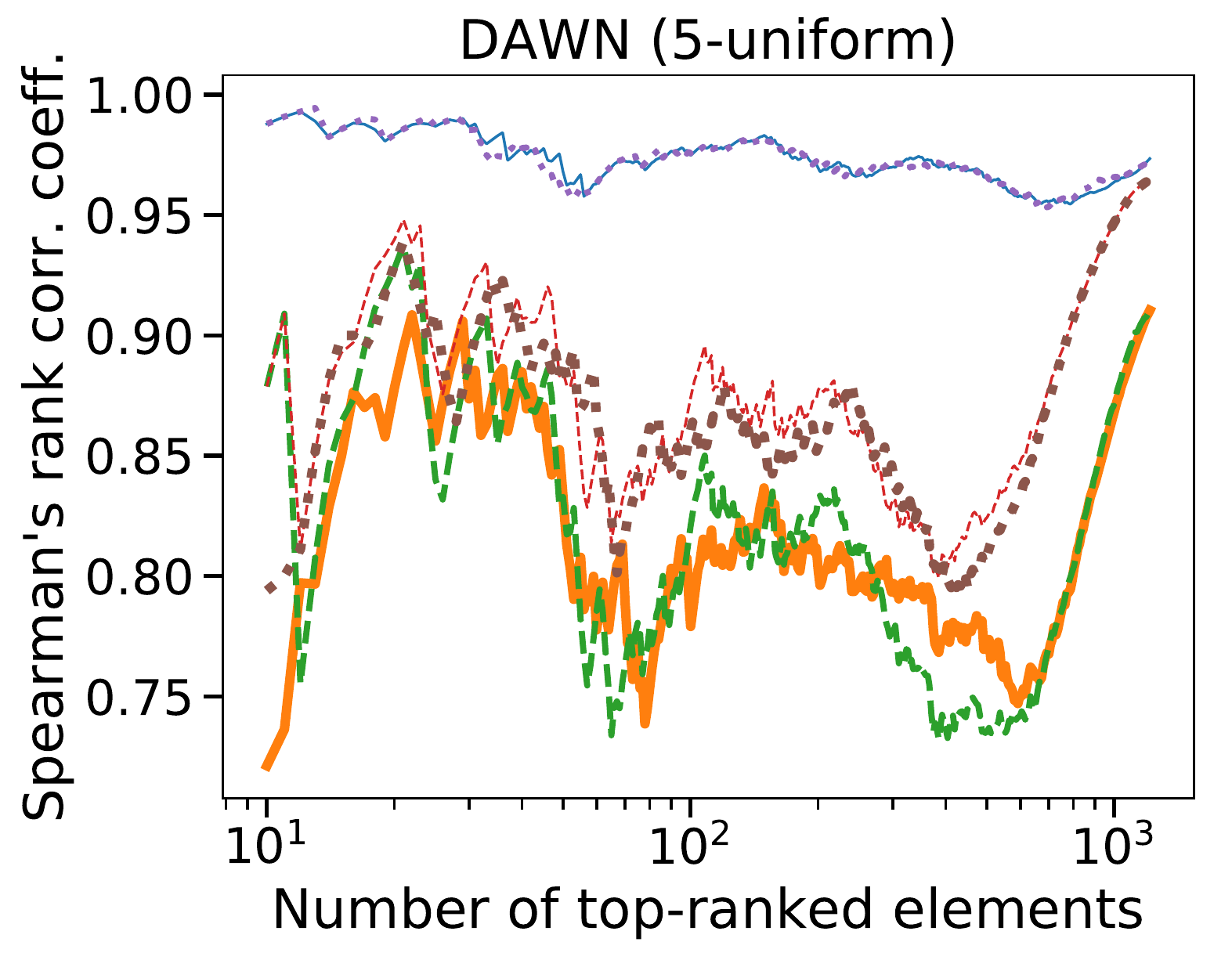}      
\captionof{figure}{Spearman's rank correlation coefficient between the top $k$ ranked
nodes from one centrality measure with the same nodes from the other two centrality
measures on the DAWN dataset (the one used to determine the nodes is listed first in the legend).
ZEC is negatively or nearly uncorrelated with CEC and HEC for the 3-uniform hypergraph,
but all centralities are quite positively correlated for the 4-uniform and 5-uniform hypergraphs,
which can also be observed from the similar top 10 nodes listed in \cref{tab:DAWN-ranks} (left).
}
\label{fig:DAWN-corrs}   
\end{minipage}
\end{figure}

\clearpage

\section{Discussion}
Centrality is a pillar of network science, and emerging datasets containing
supra-dyadic relationships offer new challenges in understanding centrality in
complex systems. Here, we proposed three eigenvector centralities for
hypergraph models of such multi-relational data. Two of these incorporated
non-linear structure and relied on fairly recent developments in the spectral
theory of tensors to create a sensible definition. None of the three centralities is
``best" and we saw empirically that the eigenvectors can provide qualitatively different results.
There are several other types of tensor eigenvectors~\cite{Qi-2017-tensor-book},
as well as other types of Perron-Frobenius theorems for hypergraph data~\cite{Michoel-2012-alignment},
which could be adapted for new centrality measures.
However, $Z$- and $H$-eigenvectors are arguably the most well-understood and commonly used tensor eigenvectors.

There are other centrality measures and ranking methods for higher-order
relational data. For example, multilinear PageRank generalizes PageRank to
tensors~\cite{Gleich-2015-multilinear-PR,Benson-2017-spacey}.
 \Citet{Tudisco-2018-multiplex} developed eigenvector centrality
for multiplex networks using new Perron-Frobenius theory of multi-linear
maps~\cite{Gautier-2017-Perron}; this is most similar to HEC.  There are also
several ranking methods for multi-relational data represented as
tensors~\cite{Kolda-2005-Higher,Kolda-2006-TOPHITS,Franz-2009-TripleRank,Ng-2011-MultiRank},
as well as notions of centrality based on simplicial complexes~\cite{Estrada-2018-simplicial}.
Finally, there are other centralitities for
hypergraphs~\cite{Kapoor-2013-centrality,Busseniers-2014-general,Bonacich-2004-hyperedges,Rodriguez-2007-functional,Estrada-2006-subgraph},
but these do not relate to the multilinear structure of tensors that we studied.

We used a set-based definition of hypergraphs that made the adjacency tensor symmetric. 
In network science, directed graphs with non-symmetric adjacency matrices
are a common model, and eigenvector centrality is still well-defined
if the graph is strongly connected. 
There are similar notions of directionality in hypergraphs. For example, the $\ngrams$
dataset could have been interpreted as ``directed'' since the ordering of the words
matters for its frequency.
Trajectory or path-based data appearing
in transportation systems~\cite{Xu-2016-HON}, citation
patterns~\cite{Rosvall-2014-memory}, and human contact
sequences~\cite{Scholtes-2017-HON} can be encoded as directed hypergraphs in
similar ways. \Cref{thm:PFZ,thm:PFH} hold for arbitrary irreducible nonnegative
tensors, which includes adjacency tensors of strongly connected hypergraphs. 
Therefore, the hypergraph centralities we
developed remain well-defined in these more general cases.
However, computation becomes a bigger challenge.

There are many choices in deciding how to construct hypergraphs from data.
As one example, we made our adjacency tensors binary (i.e., an unweighted hypergraph).
This was not necessary mathematically, and all of the proposed methods
work seamlessly if the hypergraph is weighted.
The Ask Ubuntu and DAWN datasets also demonstrated two different ways of constructing hyperedges---in the
former we included hyperedges induced by larger sets and in the latter we did not.
This choice was made to illustrate the point that there are several ways one could construct hypergraphs from data.
Our methods also relied on theory for symmetric tensors, so we studied uniform hypergraphs.
One could incorporate non-uniformity in several ways.
A simple approach could combine the scores for hypergraphs of different uniformity.
We could also ``embed'' smaller hyperedges into a larger adjacency tensor.
For example, a mixture of 3-node and 4-node hyperedges could be incorporated into an order-4 adjacency tensor,
where a 3-node hyperedge $\{i,j,k\}$ adds non-zeros in the indices that only contain $i$, $j$, and $k$ 
(e.g., setting $\cmT_{ijkk} = 1$ would create one such non-zero).
In general, hypergraphs can be a convenient abstraction, and understanding the right way of constructing
a hypergraph from data is a general research challenge.

\section*{Acknowledgments}
I thank Yang Qi and David Gleich for many helpful discussions.
I thank the reviewers for carefully reading this manuscript.
This research was supported by NSF Award DMS-1830274 and ARO Award W911NF-19-1-0057.

\bibliographystyle{bib}
\bibliography{refs}

\end{document}